\RequirePackage{amsthm}
\RequirePackage{etex}
\documentclass[sn-mathphys-num, dvipsnames]{sn-jnl}

\usepackage[utf8]{inputenc}
\usepackage[T1]{fontenc}
\usepackage[nolist]{acronym}

\usepackage{amsmath, amssymb, amsfonts, amsthm}
\usepackage[all]{xy}
\usepackage{hyperref}
\usepackage{mathtools}
\usepackage{calc}  
\usepackage{enumerate}
\usepackage{enumitem}
\usepackage{float}
\usepackage{bbm}
\usepackage{verbatim}
\usepackage{lmodern}
\usepackage{anyfontsize}

\usepackage{graphicx}
\graphicspath{{images/}{../images/}}
\usepackage{tabularx}
\usepackage{array}
\usepackage{wrapfig}
\usepackage{longtable}
\usepackage{makecell}
\usepackage{multirow, multicol}
\usepackage{ltablex}\keepXColumns

\newcolumntype{Y}{>{\centering\arraybackslash}X}
\newcolumntype{Z}{>{\scriptsize}Y}
\usepackage{rotating}

\usepackage{tikz}
\usetikzlibrary{shapes,arrows,positioning,arrows.meta,chains,positioning,quotes,fit}
\usepackage{blkarray}
\usepackage{tkz-euclide}
\usetikzlibrary{patterns} 
\usepackage{mathrsfs}
\usetikzlibrary{fadings}
\usetikzlibrary{shadows.blur}
\usetikzlibrary{spy}
\usetikzlibrary{decorations.pathreplacing}

\usepackage{nicematrix}

\usepackage{pgfplots}
\usepgfplotslibrary{colorbrewer}
\pgfplotsset{compat=1.17}

\usepackage{subfiles}
\usepackage{blindtext}
\usepackage{chngcntr}
\usepackage{todonotes}
\usepackage{diagbox}
\usepackage{autonum}

\usepackage{algorithm}
\usepackage{algpseudocode}

\usepackage[final]{pdfpages}

\usepackage{appendix}

\theoremstyle{thmstyleone}%
\newtheorem{theorem}{Theorem}
\newtheorem{proposition}[theorem]{Proposition}%
\newtheorem{corollary}[theorem]{Corollary}%
\newtheorem{lemma}[theorem]{Lemma}%

\theoremstyle{thmstyletwo}%
\newtheorem{example}{Example}%

\theoremstyle{thmstylethree}%
\newtheorem{definition}{Definition}%

\usepackage{todonotes}
\usepackage{mathtools}
\mathtoolsset{showonlyrefs}

\colorlet{review}{black} 

\newcommand{\A}{\mathcal{A}}
\newcommand{\B}{\mathcal{B}}

\newcommand{\NM}{\operatorname{NM}} 
\newcommand{\N}{\mathbb{N}} 

\newcommand{\R}{\mathbb{R}}

\newcommand{\sphere}{\mathcal{S}}

\newcommand{\Z}{\mathbb{Z}} 
\renewcommand{\Pr}{\operatorname{Prob}}

\newcommand{\F}{\ensuremath{\mathbb F}}

\newcommand{\boltzmParam}{\beta}

\DeclareMathOperator{\wt}{wt}
\DeclareMathOperator{\dist}{d}
\DeclareMathOperator{\rk}{rk}

\DeclareMathOperator{\Var}{Var}

\newcommand{\Alphabet}{\ensuremath{\mathcal{A}}}
\newcommand{\AlphabetWeight}{\ensuremath{\wt_{\mathcal{A}}}}
\newcommand{\SumAlphabetWeight}{\ensuremath{\wt_{\Sigma \mathcal{A}}}}
\newcommand{\AlphabetDist}{\ensuremath{\dist_{\Sigma \mathcal{A}}}}

\newcommand{\SumRankWeight}{\ensuremath{\wt_{\Sigma R}}}

\newcommand{\PartitionFunction}{\ensuremath{\mathcal{Z}}}

\newcommand{\card}[1]{\left\vert{#1}\right\vert} 
\newcommand{\set}[1]{\left\lbrace{#1}\right\rbrace} 
\newcommand{\st}{\, : \,} 

\definecolor{plotblue}{rgb}{0, 0.4470, 0.7410}
\definecolor{plotorange}{rgb}{0.8500, 0.3250, 0.0980}
\definecolor{plotyellow}{rgb}{0.9290, 0.6940, 0.1250}
\definecolor{plotpurple}{rgb}{0.4940, 0.1840, 0.5560}
\definecolor{plotgreen}{rgb}{0.4660, 0.6740, 0.1880}
\definecolor{plotcyan}{rgb}{0.3010, 0.7450, 0.9330}
\definecolor{plotred}{rgb}{0.6350, 0.0780, 0.1840}
\begin{acronym}
 \acro{OGF}{ordinary generating function}
\end{acronym}

\begin{document}
\title[Bounds on Sphere Sizes]{Bounds on Sphere Sizes in the Sum-Rank Metric and Coordinate-Additive Metrics}


\author*[1]{\fnm{Hugo} \sur{Beeloo-Sauerbier Couvée}}\email{hugo.sauerbier-couvee@tum.de}
\author[2]{\fnm{Thomas} \sur{Jerkovits}}\email{thomas.jerkovits@dlr.de}
\author[1]{\fnm{Jessica} \sur{Bariffi}}\email{jessica.bariffi@tum.de}

\affil[1]{\orgdiv{Department of Electrical and Computer Engineering}, \orgname{Technical University of Munich}, \orgaddress{\city{Munich}, \country{Germany}}}
\affil[2]{\orgdiv{Department of Satellite Networks}, \orgname{German Aerospace Center}, \orgaddress{\city{Wessling-Oberpfaffenhofen}, \country{Germany}}}

\abstract{
     This paper provides new bounds on the size of spheres in any coordinate-additive metric with a particular focus on improving existing bounds in the sum-rank metric. We derive improved upper and lower bounds based on the entropy of a distribution related to the Boltzmann distribution, which work for any coordinate-additive metric. Additionally, we derive new closed-form upper and lower bounds specifically for the sum-rank metric that outperform existing closed-form bounds.
}

\maketitle

\section{Introduction}\label{sec:introduction}
Classically, the study of error-correcting codes focuses on codes in the Hamming metric~\cite{hamming1950error} where the amount of errors happening depends on the number of entries that change during the transmission. With different applications, other coordinate-additive metrics were introduced and gained more attention. The sum-rank metric~\cite{nobrega}, Lee metric~\cite{lee1958some}, and also the Hamming metric, are examples of coordinate-additive metrics. Codes with distance properties in such metrics are of particular interest in various applications, such as linear network coding~\cite{Lu2005ATradeoff}, quantum-resistant cryptography~\cite{Horlemann-Trautmann2021InformationCryptography,puchinger2022generic}, coding for storage~\cite{martinezpenas2019locally}, space-time coding~\cite{shehadeh2021spacetime}.

One crucial task is to understand the performance limits of codes endowed with a given coordinate-additive metric. Bounds like the sphere-packing bound or the Gilbert--Varshamov bound~\cite{gilbert1952comparison} play a crucial role in understanding these limits, and are derived using bounds on the size of an $\ell$-dimensional ball or sphere in the corresponding metrics. The exact value for the size of an $\ell$-dimensional sphere $\sphere_{t}^{\ell}$ of radius $t$ in any coordinate-additive metric can be derived by computing all its (ordered) integer partitions, where each part of the partition has at most a part size of the maximal possible weight in the corresponding metric. These represent the decomposition of the nonzero entries of the elements in the sphere. To obtain the size of the sphere, we sum over all integer partitions, adding up the number of elements that have a weight decomposition corresponding to the integer partition. Although this procedure provides the exact value of $\card{\sphere_t^{\ell}}$, it often does not give an intuitive or practical understanding of the sphere size or how this size changes as the parameters change. For large parameters, it is even impractical to compute the size in this way.  Hence, the derivation of closed-form bounds on the exact formula is of major interest. A current method of obtaining both upper and lower bounds on $\card{\sphere_{t}^{\ell}}$ is to consider only the partition attaining the maximum number of elements. This approach is utilized by \cite{Ott2021BoundsCodes,puchinger2022generic,gruica2023densities}. Another method is to bound the size of an $\ell$-dimensional ball $\B_{t}^{\ell}$ of radius $t$, since every upper bound on $\card{\B_{t}^{\ell}}$ is also a valid upper bound on $\card{\sphere_{t}^{\ell}}$. On a complex analytic side, sizes of spheres and balls can be described using generating functions, whose coefficients can be computed using the saddle-point technique and other techniques from analytic combinatorics (see \cite{flajolet2009analytic,saddle}). In 1994, Löliger presented an information-theoretic approach for bounding the volume of an $\ell$-dimensional ball concerning any coordinate-additive metric, via the entropy of an auxiliary probability distribution~\cite{löliger1994upperbound}. Similar arguments have been used to derive bounds for Lee-metric codes \cite{bhattacharya2019method,bariffi2024error}.

Specifically addressing the sum-rank metric, closed-form upper and lower bounds for the sphere size were introduced in~\cite{puchinger2022generic,Ott2021BoundsCodes} and further discussed in \cite{gruica2023densities}. However, these bounds are limited in their tightness, particularly noticeable in scenarios involving smaller sizes of the base field $q$ and/or a larger number of blocks $\ell$.

In this paper we will discuss \textcolor{review}{different} approaches to bound sizes of spheres in a coordinate-additive metric. We derive novel improved lower and upper bounds on that size first in a general setting for any coordinate-additive metric and later explicitly focusing on the sum-rank metric.
The paper is organized as follows. We start by introducing the basic notions and concepts needed throughout the paper in Section \ref{sec:preliminaries}. In Section \ref{sec: information-theoretic} we present bounds on the size of a sphere coming from \textcolor{review}{an} information-theoretic argument. We start the section by introducing the entropy function of a random variable, as well as the probability distribution of a typical sequence in an $\ell$-dimensional sphere endowed with a given coordinate-additive metric. Furthermore, we recap existing upper bounds given by the entropy of the typical sequences. In a next step, we give novel upper and lower bounds using the entropy of a typical sequence. The bounds derived hold for any coordinate-additive metric over a given alphabet. We then turn our attention to the sum-rank metric in Section \ref{sec: Bounds for Spheres in the Sum-rank Metric} where we first recap the existing bounds on the size of an $\ell$-dimensional sphere in the sum-rank metric. Next, we present improved upper and lower bounds on the size of such a sphere using convolution arguments and ordinary generating functions, respectively. The new bounds presented in the Sections \ref{sec: information-theoretic} and \ref{sec: Bounds for Spheres in the Sum-rank Metric} are then compared to already existing bounds in Section \ref{sec:comparison}. Final conclusions are given in Section \ref{sec:conclusions}.
\section{Preliminaries}\label{sec:preliminaries}
In this section, we introduce the necessary definitions, notations, and concepts that are used throughout the paper.
In the following, let $q$ be a prime power and denote by $\F_q$ the finite field of $q$ elements. The natural numbers $\N$ include $0$. 

\subsection{Coordinate-Additive Metrics}\label{subsec:coordinate-additive}
Let $(\Alphabet,+)$ be a finite abelian group with identity element $0$ called the \textbf{alphabet}. We define a \textbf{weight function} $\AlphabetWeight : \Alphabet \to \N$ on $\Alphabet$ to be a function satisfying for all $a,b\in \Alphabet$:
\begin{enumerate}[noitemsep,topsep=4pt]
    \item $\AlphabetWeight(a) = 0$ if and only if $a=0$,
    \item $\AlphabetWeight(a) = \AlphabetWeight(-a)$,
    \item $\AlphabetWeight(a+b) \leq \AlphabetWeight(a) + \AlphabetWeight(b)$.
\end{enumerate}
This function can be extended to a \textbf{coordinate-additive weight function} on the cartesian product $\Alphabet^{\ell}$ (with group structure inherited coordinate-wise from $\Alphabet$) by defining the weight of an $\ell$-tuple to be the sum of the weights of its coordinates, i.e., $$\SumAlphabetWeight(a_1,\ldots,a_\ell) = \sum_{i=1}^{\ell}\AlphabetWeight(a_i).$$
This coordinate-additive weight function naturally induces a 
\textbf{metric} 
\begin{align}
    \begin{array}{cccc}
        \AlphabetDist : & \Alphabet^{\ell} \times \Alphabet^{\ell} & \longrightarrow & \N \\
         & (v, w) & \longmapsto & \SumAlphabetWeight(v-w).
    \end{array}
\end{align}

Examples of well-studied metrics induced by coordinate-additive weight functions are the \textbf{Hamming metric} with weight function 
\begin{align}
    \begin{array}{cccc} 
        \wt_H: & \Alphabet^{\ell} &\longrightarrow & \N \\ 
         & (a_1,\ldots,a_\ell) &\longmapsto & \sum_{i=1}^{\ell} [a_i \neq 0] 
    \end{array}
\end{align}
where $[\cdot]$  denotes the Iverson bracket and the \textbf{Lee metric} with weight function 
\begin{align}
    \begin{array}{cccc}
        \wt_L : & (\Z/n\Z)^{\ell} & \longrightarrow & \N \\ 
            & (a_1,\ldots,a_\ell) &\longmapsto & \sum_{i=1}^{\ell} \min(a_i, n-a_i),
    \end{array}
\end{align}
Of particular interest to this paper is the \textbf{sum-rank metric} on $\F_q^{m \times n}$, the space of $m \times n$ matrices over the finite field $\F_q$ with $m$ and $n$ positive integers. \textcolor{review}{In this paper, we focus on the case where $n = \eta \ell$, for two positive integers $\eta$ and $\ell$.} Then, 
under the isomorphism $\F_q^{m \times n} = \F_q^{m \times \eta \ell} \cong (\F_q^{m \times \eta})^{\ell}$, every matrix $M \in \F_q^{m \times n}$ can be represented as a sequence of $\ell$ blocks \textcolor{review}{ of $m \times \eta$} matrices $B_i \in \F_{q}^{m \times \eta}$, i.e. $M = (B_1 \mid B_2 \mid \ldots \mid B_\ell)$, and its \textbf{sum-rank weight} is given by the weight function
\begin{align}
\begin{array}{cccc}
        \SumRankWeight: & (\F_q^{m \times \eta})^{\ell} & \longrightarrow & \N \\ 
        & (B_1 \mid B_2 \mid \ldots \mid B_\ell) &\longmapsto & \sum_{i=1}^{\ell} \rk_q(B_i).
    \end{array}
\end{align}
Here $\rk_q(B_i)$ denotes the rank of $B_i$ over $\F_q$. The special case of the sum-rank metric with $\ell = 1$ is called the \textbf{rank-metric}. Also note that the cases $m = 1$ or $\eta = 1$ are examples of the Hamming metric.\medbreak

Given a coordinate-additive weight function $\SumAlphabetWeight(\cdot)$ on $\Alphabet^{\ell}$, it is important to understand the size of an $\ell$-dimensional \textbf{sphere} and \textbf{ball}, respectively, of radius $t \in \N$ defined as
\begin{align}\label{eq:sphere_and_ball}
    \sphere_t^{\ell} &:= \{v \in \Alphabet^{\ell} \ : \ \SumAlphabetWeight(v) = t\} = \SumAlphabetWeight^{-1}(t)\\
    \B_t^{\ell} &:= \{v \in \Alphabet^{\ell} \ : \ \SumAlphabetWeight(v) \leq t\} = \SumAlphabetWeight^{-1}(\{0,1,\ldots,t\}).
\end{align}
Clearly, the following identities are satisfied
\begin{align}
    \sphere_i^{\ell} \cap \sphere_j^{\ell} = \emptyset \text{ if } i \neq j, \quad \mathcal{B}_t^{\ell} = \bigcup_{i=0}^t \sphere_i^{\ell} \quad \text{and} \quad \sphere_t^{\ell} = \mathcal{B}_t^{\ell} \setminus \mathcal{B}_{t-1}^{\ell}.
\end{align}
\noindent
For the sum-rank metric we define $$\mu_{\Sigma R} := \min\{m,\eta\} \, \text{ and }\, n := \eta \ell,$$

\noindent
and denote for every $0 \leq t \leq \mu_{\Sigma R}\ell$ the \textbf{sum-rank sphere} and \textbf{sum-rank ball}, respectively, of radius $t$ by 
\begin{align}
\sphere^{m,\eta,\ell,q}_t &:= \{M \in \F_q^{m \times \eta \ell} \ : \  \SumRankWeight(M) = t\} \\
    \B^{m,\eta,\ell,q}_t &:= \{M \in \F_q^{m \times \eta \ell} \ : \  \SumRankWeight(M) \leq t\}.
\end{align}
For fixed $m$,$\eta$,$q$,$\ell$, the sum-rank sphere sizes $\big|\sphere^{m,\eta,\ell,q}_t\big|$ can be computed with a dynamic program described in~\cite{puchinger2022generic} or a closed-form described in Subsection \ref{subsec:computing_sphere_size}. 
However, this formula is generally not practical to work with when an intuitive understanding of the approximate size is needed, such as theorems involving sphere-packings or sphere-coverings. Therefore, easy-to-work-with approximations or lower/upper bounds for the sphere sizes for all $m$,$\eta$,$q$,$\ell$ are of value to research on the sum-rank metric.

\newpage

\subsection{Sizes of Spheres}\label{subsec:computing_sphere_size}
Consider an alphabet $\Alphabet$ and a coordinate-additive weight $\AlphabetWeight$. Let $\ell$ and $t$ be positive integers and $\sphere^{\ell}_{t}$ be the $\ell$-dimensional sphere of radius $t$ with respect to the weight $\AlphabetWeight$ as introduced in \eqref{eq:sphere_and_ball}. The exact size of $\sphere^{\ell}_{t}$ can be obtained by computing all the (ordered) integer partitions of $t$, where each part size is at most the maximum possible weight in the corresponding metric. More explicitly, let us define the set of all possible weights in $\Alphabet$ by $$W_\Alphabet =  \{\AlphabetWeight(a) \ : \ a \in \Alphabet\} \subset \N .$$


Let us denote the maximum possible weight in $W_\Alphabet$ by $\mu := \max_{a\in \A}\set{ \AlphabetWeight(a)}$.
We denote by $A_j$ the number of elements $a \in \Alphabet$ of weight $j \in W_\Alphabet$, i.e., $$A_j =  \card{\{a \in \Alphabet \ : \ \AlphabetWeight(a) = j\}}.$$
Furthermore, for $t,\ell \in \N$, let $\mathcal{T}^{\ell}_{t}$ denote the set of all (ordered) integer partitions of $t$ of length $\ell$ with part sizes not exceeding $\mu$, i.e.,
\begin{align}
\mathcal{T}_t^\ell := \{(t_1,t_2,\ldots,t_\ell) \in W_\Alphabet^{\, \ell} \ : \ t_1 + t_2 + \ldots + t_\ell = t \}.
\end{align}
Then, the size of the $\ell$-dimensional sphere of radius $t$ can be computed as
\begin{align}\label{eq:exact_sphere_size}
\card{\sphere_t^{\ell}} = \sum_{\substack{(t_1,\ldots, t_\ell) \in \mathcal{T}_t^\ell}} \,  \prod_{i=1}^\ell A_{t_i} .
\end{align}
This formula can further be reduced depending on the metrics considered. For instance, in the Hamming metric, since tuples of Hamming weight $t$ consist of $t$ nonzero positions, we compute the number of ways to choose $t$ positions among a total number of positions $\ell$ and multiply by the number of options we have to choose the nonzero element from, i.e.,
\begin{align}
    \card{\mathcal{S}^{\ell}_{t, \mathsf{H}}} = \binom{\ell}{t}(\card{\Alphabet} - 1)^{t}.
\end{align}
Hence, in the Hamming metric we obtain a compact closed form expression for $\card{\sphere^{\ell}_{t}}$. However, when changing the metric this task can get more involved, as the set $\mathcal{T}^{\ell}_t$ might grow rapidly and, thus, summing over all its elements would be an exhaustive task. Therefore, finding simpler upper and lower bounds on the sum or the single summands is an important task. 

For a fixed weight function $\AlphabetWeight$ and fixed integers $t$, $\ell$ the  sphere sizes $\left|\sphere^{\ell}_t\right|$ can be computed efficiently if the values of $A_j$ are known with a dynamic program that uses formula \eqref{eq:exact_sphere_size}. This is demonstrated in~\cite[Algorithm 1]{puchinger2022generic} for specifically the sum-rank metric, but the used algorithm can easily be adjusted for other coordinate-additive metrics.
Similar to \eqref{eq:exact_sphere_size}, we can calculate the size $\card{\sphere^{\ell}_{t}}$ by summing over all weight decompositions of the weight $t$ and adding the number of tuples with that weight decomposition. From an information-theoretic point of view we use the notion of a \textit{type} of a tuple $x \in \Alphabet^{\ell}$. For any positive integers $\ell$ and $t$, the type of $x \in \Alphabet^{\ell}$ is a tuple $\theta(x) = (\theta_0(x), \ldots, \theta_{\mu}(x))$ where each $\theta_i(x)$ is the relative fraction of occurrences of a weight $i \in W_{\Alphabet}$ in $x$, i.e.,
\begin{align}
    \theta_i(x) = \frac{1}{\ell}\card{\set{k \in \set{1, \ldots, \ell} \st \AlphabetWeight(x_k) = i}}.
\end{align}
Notice that the weight decomposition of $x \in \Alphabet^{\ell}$ is immediately derived from its type $\theta(x)$ by $\ell \theta(x) = (\ell \theta_0(x), \ldots, \ell \theta_{\mu}(x))$ and indeed $\AlphabetWeight(x) = \ell \sum_{i = 0}^{\mu} \theta_i(x)$. Furthermore, if $\wt(x) = t$ then $ \set{0}^{\ell \theta_0} \times \ldots \times \set{\mu}^{\ell\theta_{\mu}} \in \mathcal{T}^{\ell}_t$. Hence, to calculate $\card{\sphere_t^{\ell}}$, instead of summing over all partitions in $\mathcal{T}_t^{\ell}$ we can sum over all types in $\Alphabet^{\ell}$ yielding weight $t$. We will denote this set by $\Theta_t^{\ell}$. The number of tuples $x \in \Alphabet^{\ell}$ of type $\theta$ is given by the multinomial coefficient
\begin{align}
    \binom{\ell}{\ell \theta} = \frac{\ell !}{(\ell \theta_0)! \cdot \ldots \cdot (\ell \theta_{\mu})!}.
\end{align}
Hence, the size of the $\ell$-dimensional sphere of radius $t$ can be expressed as
\begin{align}
    \card{\sphere_{t}^{\ell}} = \sum_{\ell\theta \in \Theta^{\ell}_{t}} \binom{\ell}{\ell \theta}\textcolor{review}{.}
\end{align}
This value, however, is hard to manipulate wherefore we might use the following upper and lower bounds (see \cite[Theorem 11.1.3]{cover2006elements})
\begin{align}\label{eq:bounds_multinomial}
    \frac{1}{(\ell + 1)^{\text{}^{\card{\Alphabet} - 1}}} 2^{\ell H(\theta)} \leq \binom{\ell}{\ell\theta} \leq 2^{\ell H(\theta)},
\end{align}
where $H(\theta) := -\sum_{i : \theta_i \neq 0} \theta_i \log_2(\theta_i)$ is the entropy of $\theta$.
With this we immediately obtain upper and lower bounds on $\card{\sphere_{t}^{\ell}}$. Regarding the upper bound, further improvements can be achieved which we will discuss in Section \ref{sec: information-theoretic}.

\subsection{Ordinary Generating Functions}
The theory of \acp{OGF} is an important branch of mathematics that lays connections between combinatorics, analysis, number theory, probability theory and other fields.
In this paper we restrict ourselves to \acp{OGF} corresponding to weights in coordinate-additive metrics, which are polynomials with non-negative coefficients.
Consider a finite abelian group $\Alphabet$ with weight function $\AlphabetWeight$ and induced coordinate-additive weight function $\SumAlphabetWeight$ on $\Alphabet^{\ell}$. The {\ac{OGF} corresponding to $\SumAlphabetWeight$ is defined as the polynomial
\begin{equation}\label{def:OGF}
    \textstyle F_{\Alphabet^{\ell}}(z) := \sum_{v \in \Alphabet^{\ell}} z^{\SumAlphabetWeight(v)} = \sum_{i=0}^{\mu\ell} |\mathcal{S}_i^{\ell}| \ z^i.
\end{equation}
For a polynomial $F(z) = F_0 + F_1 z + \ldots + F_d z^d$ we use the notation $[z^i]F(z)$ to refer to the $i$-th coefficient $F_i$ of $F(z)$, with $[z^i]F(z) = 0$ for $i > \deg(F)$.
The \ac{OGF} for the sum-rank metric on $\F_q^{m \times \eta \ell}$ is denoted by \vspace{-2mm} $$\sphere^{m,\eta,\ell,q}(z) = \sum_{i=0}^{\mu\ell} \card{\sphere_i^{m,\eta,\ell,q}} \, z^i .$$

\begin{definition}[Partial order on polynomials]
Let $F(z), G(z) \in \R[z]$ be two real polynomials. If $[z^i]F(z) \leq [z^i]G(z)$ for all $i \in \N$, we say $F(z)$ is coefficient-wise less-than-or-equal to $G(z)$, denoted as
	$F(z) \preccurlyeq_c G(z).$
\end{definition}

\begin{proposition}[{\cite[Theorem I.1]{flajolet2009analytic}}]
Let $\Alphabet_1$, $\Alphabet_2$ be two finite alphabets with weight functions $\wt_{\Alphabet_1},\wt_{\Alphabet_2}$ respectively. Then $\wt_{\Alphabet_1\times\Alphabet_2}(a,b) := \wt_{\Alphabet_1}(a)+\wt_{\Alphabet_2}(b)$ is a weight function on $\Alphabet_1 \times \Alphabet_2$ and
\begin{equation}
    F_{\Alphabet_1 \times \Alphabet_2}(z) = F_{\Alphabet_1}(z) F_{\Alphabet_2}(z). 
\end{equation}
In particular, we have
\begin{equation}
	F_{\Alphabet^{\ell}}(z) = F_{\Alphabet}(z)^{\ell}, \vspace{1mm}
\end{equation}
for $\ell \in \N$. Furthermore, the product of real polynomials with non-negative coefficients preserves the partial order: If $F(z) \preccurlyeq_c G(z)$ and $K(z) \preccurlyeq_c L(z)$, then $F(z)K(z) \preccurlyeq_c G(z)L(z)$.
\end{proposition}

\begin{lemma}\label{lem: strictly increasing}
Let $F(z)$ be a real polynomial of degree $d > 0$ with non-negative coefficients $F_i \geq 0$ and first derivative $F'(z)$. If $F(z)$ is not a monomial, then the function $G(z) = z F'(z)/F(z)$ is a strictly increasing smooth function on the positive reals $\R_{> 0}$.
In particular if $F(0) > 0$, which is the case with \acp{OGF} of finite alphabets with weight functions, $G(z)$ is a bijection from $[0,\infty)$ to $[0,d)$.
\end{lemma}
\vspace{-0.5cm}
\begin{proof}
Smoothness follows directly from smoothness of $F(z)$ and $1/z$ on $\R_{>0}$. To show that $G(z)$ is strictly increasing, consider two positive integers $0 < a < b$ and let $K(a,b):= bF'(b)F(a) - aF'(a)F(b)$. Then,
    \begin{align}
        K(a,b)  & = \sum_{i=0}^{d} i F_i b^i \sum_{j=0}^{d}  F_j a^j - \sum_{i=0}^{d} i F_i a^i \sum_{j=0}^{d}  F_j b^j \\
        &= \sum_{0 \leq i < j \leq d} i F_i F_j b^i a^j  + j F_j F_i b^j a^i\ - i F_i F_j a^i b^j  - j F_j F_i a^j b^i \\
        &= \sum_{0 \leq i < j \leq d} F_i F_j (j-i)( b^j a^i  -  a^j b^i) > 0
    \end{align}
    where the last inequality follows from assuming that $F(z)$ is not a monomial,  so there exist $i < j$ such that $F_i F_j > 0$. The proof now follows since $G(b) - G(a) = \frac{K(a,b)}{F(a)F(b)} > 0.$
Lastly, we have  that $\lim_{z \to \infty} F'(z)/z^{d-1} = d F_d$ and $\lim_{z \to \infty} F(z)/z^{d} = F_d$, hence we get $\lim_{z \to \infty} G(z) = d$.
\end{proof}


\section{Information-Theoretic Bounds on Spheres}\label{sec: information-theoretic}
In \cite{löliger1994upperbound} an asymptotically tight upper bound on the volume of an $\ell$-dimensional ball $\card{\B_{t}^{\ell}}$ of radius $t$ was introduced. The author used information theoretic tools, as described in Section \ref{subsec:computing_sphere_size}, and bounded the size of the sphere via the type maximizing the entropy, and showed that this bound is asymptotically tight. Furthermore, this bound is valid for any arbitrary additive weight function $\AlphabetWeight$ (see Section~\ref{subsec:coordinate-additive} for the definition) with respect to some finite abelian group $\Alphabet$.
The bound was proved to hold for normalized weights $\rho$ with $\rho :=  t/\ell$ up to the average weight $$ \overline{w} := |\Alphabet|^{-1} \sum_{a \in \Alphabet} \AlphabetWeight(a)$$ at which the volume of the ball is saturated.
We extend the result from~\cite{löliger1994upperbound} to the size of spheres and also prove that the bound holds for $\rho\geq\overline{w}$ up to the maximum possible weight, i.e., for $0 < \rho < \mu = \max_{a\in \A}\set{ \AlphabetWeight(a)}$. 

Given a random variable $X$ over a finite alphabet $\Alphabet$ with probability distribution $P$, we define $P(a) := \Pr(X=a)$ with $a \in \Alphabet$.
The entropy $H(P)$ of $P$ with respect to the base $q$ is defined as
\begin{align}\label{eq:entropy_q}
    H(P) := -\sum_{a \in \Alphabet,  P(a) \neq 0} P(a)\log_{q} P(a).
\end{align}
\begin{example}\vspace{-3mm}
Consider an alphabet $\Alphabet = \F_q$, a probability $\rho \in (0,1)$, and a probability distribution
\begin{align}
    P(x) = 
    \begin{cases}
        1 - \rho & x = 0\\
        \rho/(q-1) & x \neq 0
    \end{cases}.
\end{align}
\vspace{0.3cm}
Then the entropy $H(X)$ reduces to the $q$-ary entropy function defined as
$$h_{q}(\rho) := - \rho  \log_{q} (\rho/(q-1)) -  (1-\rho)\log_{q} (1-\rho).$$
\end{example}

\begin{definition}\label{def: boltzmann distribution}
For any $a \in \Alphabet$ and $0 < \rho < \mu$, we define the probability distribution
\begin{align}\label{equ:boltzmann_distribution}
    P_{\beta}(a) := \frac{q^{-\boltzmParam \AlphabetWeight(a)}}{\PartitionFunction(\boltzmParam)}
\end{align}
where $\boltzmParam$ is the unique solution to the weight constraint
\vspace{0.1cm}
\begin{align}\label{eq:weightconstraint}
    \mathbb{E}[\AlphabetWeight(a) ] = \sum_{a \in \Alphabet} P_{\beta}(a)\AlphabetWeight(a) = \rho
\end{align}
and $\PartitionFunction(\beta)$ is chosen s.t. $\sum_{a \in \Alphabet} P_{\beta}(a)=1$, i.e. $\PartitionFunction(\beta) = \sum_{a \in \Alphabet} q^{-\beta\AlphabetWeight(a)}.$
\end{definition}

Later, in Proposition \ref{prop: rho beta bijection}, we show that for any $0 < \rho < \mu$\textcolor{review}{,} a  solution $\beta$ to the equation $\sum_{a \in \Alphabet} P_{\beta}(a)\AlphabetWeight(a) = \rho$ always exists and is unique. This establishes a one-to-one correspondence between $\rho$ and $\beta$, i.e. every $\beta \in \mathbb{R}$ determines a $\rho \in (0,\mu)$ via the function 
\begin{equation}\label{eq: rho(beta)}
  \rho(\beta) :=  \sum_{a \in \Alphabet} P_{\beta}(a)\AlphabetWeight(a)
\end{equation}
and vice versa.
Hence we will use $\rho$ and $\rho(\beta)$ interchangeably.
Let us denote by $$H_{\rho} := H(P_{\beta})$$ the entropy of the distribution in \eqref{equ:boltzmann_distribution}. Then, the following bound was proven in \cite{löliger1994upperbound}.
\begin{theorem}[\cite{löliger1994upperbound}]\label{thm:bound_loeliger}
For any  $0 < \rho \leq  \overline{w}$ and $\ell \in \mathbb{N}$ we have
    \begin{equation}
        \frac{1}{\ell} \log_q \card{\B_{\rho \ell}^{\ell}} \leq H_{\rho}.
    \end{equation}
\end{theorem}
The following is an immediate consequence of Theorem \ref{thm:bound_loeliger} above.
\begin{corollary}\label{cor:bound_loeliger}
    For any  $0 < \rho < \overline{w}$ and $\ell \in \mathbb{N}$ we have 
   
    \begin{equation}
        \frac{1}{\ell} \log_q \card{\sphere_{\rho \ell}^{\ell}} \leq H_{\rho}.
    \end{equation}
\end{corollary}
By a simple cutting argument, the bound in Theorem \ref{thm:bound_loeliger} can be extended to any normalized weight $0 \leq \rho \leq \max_{x\in \A}\set{ \AlphabetWeight(x)}=: \mu$ as follows.
\begin{theorem}
    For any  $0 < \rho \leq  \mu$ we have 
    \begin{align}
        \frac{1}{\ell} \log_q \card{\mathcal{B}_{\rho \ell}^{\ell}} 
        \leq
        \begin{cases}
            H_{\rho} & \text{if } 0 < \rho \leq  \overline{w}\\
            \log_q(\card{\A}) & \text{if } \overline{w} < \rho \leq  \mu
        \end{cases}.
    \end{align}
\end{theorem}

\subsection{Upper Bound}
In this subsection, we show that Corollary \ref{cor:bound_loeliger} also holds for normalized weights $0 < \rho < \mu$.
Recall the \acp{OGF} for the alphabets $\A$ and $\A^{\ell}$, respectively,
\begin{align}
    F_\Alphabet(z) = \sum_{a \in \Alphabet} z^{\AlphabetWeight(a)}\quad \text{and } \quad F_{\Alphabet^\ell} (z) = \sum_{v \in \Alphabet^\ell} z^{\SumAlphabetWeight(v)} = F_\Alphabet(z)^\ell .
\end{align}
We now can express $\PartitionFunction(\beta)$, $\rho(\beta)$ and $H_{\rho}$ in terms of $F_\Alphabet(z)$.
\begin{lemma}\label{lem: in terms of generating function}
    Let  $\beta \in \R$ and $\rho = \rho(\beta)$. Then 
    \begin{align}\label{eq:rho(beta)}
        \PartitionFunction(\beta) = F_\Alphabet\left(q^{-\beta}\right), \quad
        \rho(\beta) = q^{-\beta} \frac{F'_\Alphabet\left(q^{-\beta}\right)}{F_\Alphabet\left(q^{-\beta}\right)}, \quad
        H_{\rho} = \log_q\left(\frac{F_\Alphabet(q^{-\beta})}{q^{-\beta\rho(\beta)}}\right).
    \end{align}
\end{lemma}
\vspace{-5mm}
\begin{proof}
    The first equality immediately follows from the fact that $\PartitionFunction(\boltzmParam)$ is the normalization constant of the distribution $P_\beta(\cdot)$: $\PartitionFunction(\boltzmParam) = \sum_{a \in \Alphabet} \left(q^{-\beta}\right)^{\AlphabetWeight(a)} = F_{\Alphabet}(q^{-\beta})$.
    
    For the second equality, we rewrite the weight constraint \eqref{eq:weightconstraint} defining $\rho(\beta)$ to obtain 
    \begin{align}
        \rho(\beta) &= \sum_{a \in \Alphabet} P_{\beta}(a) \AlphabetWeight(a)= \sum_{a \in \Alphabet} \frac{q^{-\beta \AlphabetWeight(a)}}{\PartitionFunction(\beta)} \AlphabetWeight(a)
            = \sum_{a \in \Alphabet} \frac{q^{-\beta (\AlphabetWeight(a) - 1)} q^{-\beta} \AlphabetWeight(a)}{F_{\Alphabet}(q^{-\beta})}
    \end{align}
    Noticing that $F'_\Alphabet\left(q^{-\beta}\right) = \sum_{a \in \Alphabet} q^{-\beta (\AlphabetWeight(a) - 1)}\AlphabetWeight(a)$ yields the result.

    To show the last equality we use equations \eqref{equ:boltzmann_distribution}, \eqref{eq: rho(beta)} and the previous equalities for $\PartitionFunction(\beta)$. Rewriting the definition of the entropy $H_{\rho}$ gives
\vspace{-3mm}
    \begin{align}
        H_{\rho} 
        &= -\sum_{a \in \Alphabet} P_{\beta}(a) \log_q \left( P_{\beta}(a) \right) = -\sum_{a \in \Alphabet} P_{\beta}(a) \log_q \left( \frac{q^{-\beta\AlphabetWeight(a)}}{F_{\Alphabet}(q^{-\beta})} \right)\\
        &= \sum_{a \in \Alphabet} P_{\beta}(a)\left(\beta \AlphabetWeight(a) + \log_q\left( F_{\Alphabet}(q^{-\beta}) \right)\right) \\
         &= \beta \sum_{a \in \Alphabet} P_{\beta}(a) \AlphabetWeight(a) \ + \ \log_q\left( F_{\Alphabet}(q^{-\beta}) \right)\sum_{a \in \Alphabet} P_{\beta}(a)\\
       	& = \beta\rho(\beta) + \log_q\left(F_{\Alphabet}(q^{-\beta})\right)\\
        &= \log_q\left( \frac{F_\Alphabet(q^{-\beta})}{q^{-\beta\rho(\beta)}} \right).
    \end{align}
\end{proof}

\begin{proposition}\label{prop: rho beta bijection}
Let $\Alphabet$ be an alphabet with weight function $\AlphabetWeight$. Then the function $\rho(\beta)$ given in~\eqref{eq: rho(beta)} is a strictly decreasing smooth bijection from $\mathbb{R}$ to $(0,\mu)$.   
\end{proposition}
\vspace{-8mm}
\begin{proof}
By Lemma \ref{lem: strictly increasing}, the function $G(z) := z \frac{F'_\Alphabet\left(z\right)}{F_\Alphabet\left(z\right)}$  is a strictly increasing bijection from $[0,\infty)$ to $[0,\mu)$. The proof now follows since $\rho(\beta) = G(q^{-\beta})$. 
\end{proof}

We will now \textcolor{review}{use the technique of bounding the coefficients of an \ac{OGF} by evaluating it at points $y$.}
For any real valued $y > 0$ we have
\begin{align}\label{eq:saddle-point}
    |\sphere_t^\ell| y^t &= \left([z^t]F_{\Alphabet^\ell}(z)\right) y^t \leq F_{\Alphabet^\ell}(y) = F_\Alphabet(y)^\ell.
\end{align}
We can further rewrite this expression and take the infimum on the right-hand side and obtain \vspace{-2mm}
\begin{align}\label{eq:saddle-point-alt}
    \frac{1}{\ell} \log_q |\sphere_t^\ell|  &\leq \inf_{y > 0}\ \log_q\left(\frac{F_\Alphabet(y)}{y^\rho}\right).
\end{align}
Moreover, we can show that a global minimum of $F_\Alphabet(y)/y^\rho$ exists and therefore the infimum is a minimum: by setting the derivative of $F_\Alphabet(y)/y^\rho$ to zero and using \eqref{eq:rho(beta)} for $\rho$, we obtain a local minimum for $y = q^{-\beta}$. Then, using Lemma~\ref{lem: strictly increasing}, we can show that the derivative of $F_\Alphabet(y)/y^\rho$ is negative for $0 < y < q^{-\beta}$ and positive for $y > q^{-\beta}$.
Therefore, the local minimum is also the global minimum, where the function $\log_q\left(\frac{F_\Alphabet(y)}{y^\rho}\right)$ takes the value $H_\rho$ (cf. \eqref{eq:rho(beta)}).  

\textcolor{review}{This technique is similar to the more general \textit{saddle-point bound} explained in~\cite[Section VIII.2]{flajolet2009analytic}.}
To summarize, the saddle-point-like bound \eqref{eq:saddle-point-alt} coincides with the entropy bound (see also~\cite[Theorem 4.1]{GREFERATH200411},~\cite[Theorem IV.9]{byrne2021fundamental}), but extends the range of $\rho$ to $(0,\mu)$, as stated in Theorem \ref{thm:upperbound_sphere_saddlepoint}.
Finally, we have 
\begin{align}
    \frac{1}{\ell} \log_q |\sphere_t^\ell|  &\leq \log_q\left(\frac{F_\Alphabet(q^{-\beta})}{(q^{-\beta})^\rho}\right).
\end{align}

\begin{theorem}\label{thm:upperbound_sphere_saddlepoint}
    For any  $0 < \rho < \mu$ and $\ell \in \mathbb{N}$ we have
    \begin{equation}
        \frac{1}{\ell} \log_q |\sphere_{\rho \ell}^\ell| \, \leq \, H_{\rho} \, = \, \inf_{y > 0}\ \log_q\left(\frac{F_\Alphabet(y)}{y^\rho}\right),
    \end{equation}
    where the infimum (in this case a minimum) on the right-hand side is attained at $y = q^{-\beta}$, with $\beta$ defined by $\rho(\beta) = \rho$.
\end{theorem}

\subsection{Lower Bounds}
\noindent
Observe that \eqref{eq:bounds_multinomial} gives us the lower bound
\begin{equation}\label{eq:types_lower_bound}
    \frac{1}{\ell} \log_q |\sphere_{\rho \ell}^\ell| \geq H_\rho - (\card{\Alphabet}-1)\frac{1}{\ell}\log_q(\ell+1).
\end{equation}
Although this bound is asymptotically tight as $\ell \to \infty$, it can be improved for some ranges of $\ell$ when the alphabet size $\card{\Alphabet}$ is large (e.g. for the sum-rank metric where $\card{\Alphabet} = q^{mn}$). In this section we derive a different lower bound that does not depend on $\card{\Alphabet}$ but instead on $\AlphabetWeight$. Specifically for the sum-rank metric this new bound is tighter than \eqref{eq:types_lower_bound} for a large range of $\ell$.

Let $X_{\beta}, X_{\beta,1}, X_{\beta,2},X_{\beta,3},\ldots$ be \textcolor{review}{independent and identically distributed} random variables taking values in $\Alphabet$ with probability distribution $P_{\beta}$ given in \eqref{equ:boltzmann_distribution}. For $a \in \Alphabet$, define the function
\begin{align}\label{eq:phi_beta}
   \qquad\qquad \varphi_\beta(a) := -\log_q\left(P_{\beta}(a)\right) \  = \beta \AlphabetWeight(a) + \log_q\left( F_{\Alphabet}(q^{-\beta}) \right).
\end{align}
Note that $H_{\rho} = \mathbb{E}[\varphi_\beta(X_{\beta,i})]$ with $\rho = \rho(\beta)$ for each $i \in \N$. 
As a consequence of Chebyshev’s inequality \cite{tchebichef1874valeurs}, for any $\gamma > 0$ it holds
\begin{align}\label{eq:chebyshev}
    \Pr\left(\left|\frac{1}{\ell}\sum_{i=1}^\ell \varphi_\beta(X_{\beta,i}) - H_{\rho} \right| \geq \gamma \right) \leq \frac{\Var(\varphi_\beta(X_{\beta}))}{\ell \gamma^2} = \frac{\beta^2 \Var(\AlphabetWeight(X_{\beta}))}{\ell \gamma^2}.
\end{align}

By choosing $\gamma$ appropriately and by following similar techniques as used in \cite{löliger1994upperbound}, a lower bound for a sum of sphere sizes is derived in Theorem \ref{thm: lower bound sum spheres}.
\vspace{-3mm}
\begin{theorem}\label{thm: lower bound sum spheres}
Given $t = \rho \ell$ and $0 < \varepsilon < 1$, let $\beta$ be defined by the weight constraint \eqref{eq:weightconstraint}. Furthermore, let $\delta = \left(\frac{\ell \Var( \AlphabetWeight(X_{\beta}))}{(1-\varepsilon)}\right)^{1/2}$. Then 
\begin{equation}
    \sum_{\substack{-\delta < j < \delta\\ j \in \Z}} |\sphere_{t + j}^\ell| 
    \ > \
    \varepsilon \ q^{\ell H_\rho - |\beta| \delta}.
\end{equation}
\end{theorem}
\vspace{-7mm}
\begin{proof}\renewcommand{\qedsymbol}{}
    First, let us define a probability distribution on $\Alphabet^\ell$ naturally by $P_\beta(v) := \prod_{i=1}^\ell P_\beta(v_i)$ for $v = (v_1,\ldots,v_\ell) \in \Alphabet^\ell$. Let $Y_\beta$ denote a random variable taking values in $\A^\ell$ with this probability distribution. Note that $Y_\beta$ is identically distributed as the tuple $(X_{\beta,1},\ldots,X_{\beta,\ell})$.
    Next, for any $v \in \Alphabet^\ell$, the expression $\card{\SumAlphabetWeight(v) - \rho \ell} < \delta$ is equivalent to \vspace{-2mm}
    \begin{align}
        \card{\frac{1}{\ell}\beta\SumAlphabetWeight(v) - \beta \rho } < \frac{|\beta| \delta}{\ell}.
    \end{align}
    Note that, using Lemma \ref{lem: in terms of generating function} and \eqref{eq:phi_beta}, we rewrite the left hand side of the inequality as \vspace{-2mm}
    \begin{align}
        \card{\frac{1}{\ell}\beta\SumAlphabetWeight(v) - \beta \rho } 
        =
        \card{\frac{1}{\ell}\sum_{i=1}^{\ell}\varphi_\beta(v_i) -  H_\rho}.
    \end{align}
   Hence, we have \vspace{-2mm}
    \begin{align}
        \Pr\left(\card{\SumAlphabetWeight(Y_\beta) - \rho \ell} < \delta\right) = \Pr\left(\left|\frac{1}{\ell}\sum_{i=1}^\ell \varphi_\beta(X_{\beta,i}) \ - \ H_{\rho} \right|  < \frac{|\beta| \delta}{\ell} \right).
    \end{align}
    Applying Chebyshev's inequality (see \eqref{eq:chebyshev}), we obtain a lower bound
    \begin{align}
       \Pr\left(\card{\SumAlphabetWeight(Y_\beta) - \rho \ell} < \delta\right) \geq 1 - \frac{\ell \Var(\AlphabetWeight(X_{\beta}))}{ \delta^2} = \varepsilon .
    \end{align}
    Now, let us denote $U := \displaystyle{\bigcup_{\substack{-\delta < j < \delta \\ j \in \Z}} } \sphere_{\rho\ell+j}^\ell$ and note that we can write
    \begin{align}
         \Pr\left(\card{\SumAlphabetWeight(Y_\beta) - \rho \ell} < \delta\right)  
         = \Pr\left( Y_\beta \in U  \right) 
         = \sum_{v  \in  U}  P_\beta(v).
    \end{align}
     Furthermore, note that the inequality $\card{\frac{1}{\ell}\sum_{i=1}^{\ell}\varphi_\beta(v_i) -  H_\rho} < \frac{|\beta| \delta}{\ell}$ implies that $q^{-\ell H_{\rho} + |\beta|\delta} > P_\beta(v)  $. 
    The proof is completed observing that
    \begin{equation}
        \text{}\hspace{17mm} \sum_{\substack{-\delta < j < \delta \\ j \in \Z}}  \card{\sphere_{\rho\ell+j}^\ell}   q^{-\ell H_{\rho} + |\beta|\delta} =  \ \card{U}   q^{-\ell H_{\rho} + |\beta|\delta} > \sum_{v  \in  U}  P_\beta(v) \geq \varepsilon.  \hspace{17mm} \square
    \end{equation}
\end{proof}
 Using the inequality
$
  {\displaystyle{  \max_{\substack{-\delta < j < \delta\\ j \in \Z}} }}|\sphere_{t + j}^\ell| \ \geq \ \frac{1}{2\lceil\delta\rceil - 1}{\displaystyle{ \sum_{\substack{-\delta < j < \delta\\ j \in \Z}} }} |\sphere_{t + j}^\ell| $ we get an alternative bound in Theorem \ref{thm:entropy LB 2}.
\begin{theorem}\label{thm:entropy LB 2}
Given $t = \ell \rho$ and $0 < \varepsilon < 1$, let $\beta$ be defined by the weight constraint \eqref{eq:weightconstraint} and $\delta= \left(\frac{\ell \Var( \AlphabetWeight(X_{\beta}))}{(1-\varepsilon)}\right)^{1/2}$. Then 
\begin{equation}
    \max_{\substack{-\delta < j < \delta\\ j \in \Z}} \ \frac{1}{\ell} \log_q |\sphere_{t + j}^\ell| 
   \  > \  H_{\rho} - \frac{|\beta| \delta}{\ell} - \frac{1}{\ell} \log_q \left(
    \frac{2\lceil\delta\rceil - 1}{\varepsilon}\right) . 
\end{equation}
\end{theorem}

Empirically, good bounds seem to be obtained for $\varepsilon$ close to 0 (this corresponds to a relatively small $\delta$ close to $\ell^{1/2} \Var(\AlphabetWeight(X_{\beta}))^{1/2}$). There exists a number $ 0 < \varepsilon_0 < 1$ such that setting $\varepsilon = \varepsilon_0$ implies $\delta  = \lfloor \ell^{1/2} \Var(\AlphabetWeight(X_{\beta}))^{1/2} \rfloor + 1$; we suspect that the best bounds are obtained for $\varepsilon$ close to this $\varepsilon_0$.

 Moreover, for constant $\varepsilon$ and $\rho$, this bound coincides asymptotically with the bound from Theorem \ref{thm:upperbound_sphere_saddlepoint} as $\ell \to \infty$  and is therefore asymptotically tight.
\section{Bounds on Spheres in the Sum-rank Metric}\label{sec: Bounds for Spheres in the Sum-rank Metric}

In this section we derive improved closed-form upper and lower bounds on the size of a sphere in the sum-rank metric. Hence we fix $m$, $\eta$ and $q$ and we use $\NM_q(m,\eta, t)$ to denote the number of matrices of rank $t$ over $\F_q^{m \times \eta}$.
Following~\eqref{eq:exact_sphere_size}, an exact formula for the size of the $\ell$-dimensional sphere of radius $t$ in the sum-rank metric is given by
\begin{align}\label{eq: exact_sphere_size_sum-rank}
\card{\sphere^{m,\eta,\ell,q}_t} = \sum_{\substack{(t_1,\ldots, t_\ell) \in \mathcal{T}_t^\ell}} \,  \prod_{i=1}^\ell \NM_q(m,\eta, t_{i}),
\end{align}
with $\mathcal{T}^{\ell}_{t}$  the set of all (ordered) integer partitions of $t$ of length $\ell$ with part sizes not exceeding $\mu = \mu_{\Sigma R} = \min\{m,\eta\}$, i.e.,
\begin{align}
\mathcal{T}_t^\ell := \left\lbrace(t_1,t_2,\ldots,t_\ell) \in \{0,1,\ldots,\mu\}^\ell \ : \ t_1 + t_2 + \ldots + t_\ell = t \right\rbrace.
\end{align}
A current method
of obtaining both closed-form upper and lower bounds on $\card{\sphere^{m,\eta,\ell,q}_t}$ uses \textcolor{review}{the following} steps:\\
\begin{enumerate}
\itemsep0.5em
    \item Derive a closed-form upper bound $b^{\text{up}}_q(m,\eta, t)$ and lower bound $b^{\text{low}}_q(m,\eta, t)$
    on  $\NM_q(m,\eta, t)$;
    \item Find a partition $t^{\text{up}} = (t_1^{\text{up}},\ldots,t_\ell^{\text{up}}) \in \mathcal{T}_t^\ell$ that maximizes $\prod_{i=1}^\ell b^{\text{up}}_q(m,\eta, t_{i}^{\text{up}})$ and a partition $t^{\text{low}} = (t_1^{\text{low}},\ldots,t_\ell^{\text{low}}) \in \mathcal{T}_t^\ell$ that maximizes $\prod_{i=1}^\ell b^{\text{low}}_q(m,\eta, t_{i}^{\text{low}})$, for example with the use of Lagrange multipliers;
    \item Derive a closed-form upper bound $N_t^\ell$ on $\card{\mathcal{T}_t^{\ell}}$.
\end{enumerate}
\text{}\\
This strategy yields the following general bounds on the size of a sphere:
\begin{align}\label{eq: current general bounds sum-rank}
   \prod_{i=1}^\ell b^{\text{low}}_q(m,\eta, t_{i}^{\text{low}}) \, \leq \,  \card{\sphere^{m,\eta,\ell,q}_t} \,\leq\, N_t^\ell \prod_{i=1}^\ell b^{\text{up}}_q(m,\eta, t_{i}^{\text{up}}).
\end{align}
In \cite{Ott2021BoundsCodes}, and further discussed in \cite{gruica2023densities}, the authors used  $b^{\text{low}}_q(m,\eta, t) = \gamma_q^{-1} q^{t(m+\eta-t)}$ 
to obtain the closed-form lower bounds
\begin{align}\label{eq: lower  bound ott}
\gamma_q^{-\ell} q^{t(m+\eta-\frac{t}{\ell}) - \frac{\ell}{4}} \, \leq \, \gamma_q^{-\ell} \binom{\ell}{r}q^{t(m+\eta - \frac{t}{\ell})+\frac{r^2}{\ell}-r}  \leq  \card{\sphere^{m,\eta,\ell,q}_t},
\end{align}
where $r$ satisfies $t \equiv r \mod \ell$ with $0  \leq r < \ell$ and
$\gamma_q$ is a special instance of the $\boldsymbol{q}$-\textbf{Pochhammer symbol}  defined as 
\begin{align}
   (a;x)_k := \prod_{i=0}^{k-1} (1-a  x^i), \quad (a;x)_\infty := \prod_{i=0}^\infty (1-a  x^i), \quad \gamma_q := \left(\frac{1}{q};\frac{1}{q}\right)_{\infty}^{-1} = \prod_{i=1}^\infty (1-q^{-i})^{-1}.
\end{align}
Likewise in \cite{puchinger2022generic} this strategy is used with
$b^{\text{up}}_q(m,\eta, t) = \gamma_q q^{t(m+\eta-t)}$ and $N_t^\ell = \binom{\ell + t - 1}{\ell-1}$ to obtain the closed-form upper bounds
\begin{align}\label{eq: upper  bound sven}
\card{\sphere^{m,\eta,\ell,q}_t} \textcolor{review}{ \leq \gamma_q^\ell \binom{\ell + t - 1}{\ell - 1} q^{t(m+\eta-\frac{t}{\ell})+\frac{r^2}{\ell}-r}} \leq \gamma_q^\ell \binom{\ell + t - 1}{\ell - 1} q^{t(m+\eta-\frac{t}{\ell})}.
\end{align}
\noindent
\textcolor{review}{where $r$ satisfies $t \equiv r \mod \ell$ with $0  \leq r < \ell$.} Note that as  $\ell$ increases, the discrepancy between the lower bounds of \eqref{eq: lower  bound ott} and upper bounds of \eqref{eq: upper  bound sven} also greatly increases. Hence in the regime of large $\ell$ better bounds can be obtained. 

First, we reuse the current strategy \eqref{eq: current general bounds sum-rank} with a different upper bound $b^{\text{up}}_q(m,\eta, t)$ on $\NM_q(m,\eta, t)$ that works better for large $\ell$ and small $t$ in Section \ref{subsec: bounds on matrices}. We also derive an improved lower bound that we will use later in Section \ref{subsec: lower bound OGF}. 

Secondly, in Section \ref{subsec: upper bound convolution}, we investigate a different strategy where, instead of estimating the sum in \eqref{eq: exact_sphere_size_sum-rank} directly, we exploit the  nature of the sum-rank sphere size as an $\ell$-fold discrete convolution, and bound it
with continuous convolutions for which closed-form expression exist.


Lastly, in Section \ref{subsec: lower bound OGF}, we utilize  
$\sphere^{m,\eta, \ell, q}(z) = \left(\sphere^{m,\eta, 1, q}(z)\right)^\ell$, where we derive an alternative function that serves as a coefficient-wise lower bound for $\sphere^{m,\eta, 1, q}(z)$.
The $\ell$-th power of this function is calculated more efficiently, since its roots are expressed easily.




\subsection{Bounds on the Number of Matrices of Fixed Rank}\label{subsec: bounds on matrices}

For $a, b \in \N$ we define the $\boldsymbol{q}$-\textbf{binomial coefficient} as 
\begin{align}
    \begin{bmatrix}a \\ b \end{bmatrix}_{q} = \prod_{i = 1}^{b}\frac{1-q^{a-b+i}}{1-q^{i}} = \prod_{j = 0}^{b-1}\frac{q^a-q^{j}}{q^b-q^{j}}.
\end{align}
Then, by~\cite{migler2004weight}, the number of matrices of rank $t$ over $\F_q^{m \times \eta}$ is computed as 
\begin{align}
    \NM_q(m,\eta, t) = \begin{bmatrix}m\\t\end{bmatrix}_q \ \prod_{i=0}^{t-1} (q^\eta - q^i).
\end{align}

Let $q \geq 2$, $\mu = \mu_{\Sigma R} = \min\{m,\eta\}$ and $\mathfrak{M} := \max\{m,\eta\}$. We will lower-bound $\NM_q(m,\eta, t)$ in two steps, eventually leading to Proposition \ref{prop: new gamma LB number matrices} below; first, the $q$-binomial coefficients satisfy a useful lower bound  that follows from elementary arguments (see~\cite[Lemma 2.2]{ihringer2015phdthesis}):
\begin{align}\label{eq: ihringer}
\begin{bmatrix}\mu\\i\end{bmatrix}_q 
    &\geq 
    \begin{cases}
        (1+\frac{1}{q})q^{i(\mu-i)} & \text{if $0 < i < \mu$}\\
	1 & \text{if $i \in \{0,\mu\}$}\\
    \end{cases}.
\end{align} 

\begin{lemma}\label{lem: q-poch upper bound}
Let $0 < i < \mu$ and  $q \geq 2$. Then
\[
\begin{bmatrix}
\mu\\i\end{bmatrix}_{1/q^2} \leq \left(\frac{1}{q^2};\frac{1}{q^2}\right)_\infty^{-1} \leq 1+\frac{1}{q}.
\]
\end{lemma}
\vspace{-0.5cm}
\begin{proof}\renewcommand{\qedsymbol}{}
By Euler's pentagonal number theorem  \cite[Thm 15.5]{van2001course} we can rewrite and bound the $q$-Pochhammer symbol as follows:
\begin{align}
\left(\frac{1}{q^2};\frac{1}{q^2}\right)_\infty &= 1 + \sum_{n=1}^{\infty} (-1)^n\left[\left(\frac{1}{q}\right)^{3n^2-n}+\left(\frac{1}{q}\right)^{3n^2+n}\right]\\
&= 1 - \left(\frac{1}{q}\right)^{2} - \left(\frac{1}{q}\right)^{4} + \left(\frac{1}{q}\right)^{10} + \left(\frac{1}{q}\right)^{14} - \ldots \\
&\geq 1 - \left(\frac{1}{q}\right)^{2} - \left(\frac{1}{q}\right)^{4}\\
&\hspace{-0.5cm}\stackrel{\text{(for $q \geq 2$)}}{\geq} 1 - \frac{1}{q^2 - 1} \geq 1 - \frac{1}{q+1} = \frac{q}{q+1} .
\end{align}
So for $0 < i < \mu$ we have
\begin{align}
\text{}\hspace{10mm}\begin{bmatrix}
    \mu\\i\end{bmatrix}_{1/q^2} 
    =  \frac{ \left(\frac{1}{q^2};\frac{1}{q^2}\right)_\mu}{\left(\frac{1}{q^2};\frac{1}{q^2}\right)_i \left(\frac{1}{q^2};\frac{1}{q^2}\right)_{\mu-i}} 
    \leq \frac{1}{\left(\frac{1}{q^2};\frac{1}{q^2}\right)_i}
    \leq \frac{1}{\left(\frac{1}{q^2};\frac{1}{q^2}\right)_\infty} 
    \leq 1+\frac{1}{q}.\hspace{10mm} \square
\end{align}

\end{proof}

\noindent
As direct corollary of  \eqref{eq: ihringer} and Lemma \ref{lem: q-poch upper bound} we obtain for all $0 \leq i \leq \mu$ the inequality 
\begin{align}\label{eq: 1q2 leq q}
    \begin{bmatrix}\mu\\i\end{bmatrix}_{1/q^2} q^{i(\mu - i)}  
        \leq 
    \begin{bmatrix}\mu\\i\end{bmatrix}_{q}.
\end{align}
For $i\in \{0,\mu\}$ this follows trivially as the $q$-binomial coefficients equal 1 then.

Secondly, for $a,b,c \in \mathbb{N}$ with $0 \leq a < b < c$ we have 

\begin{align}
    \left( \prod_{j=0}^{b-1} (q^c - q^j) \right)^{a} &= \left( \prod_{j=0}^{a-1} (q^c - q^j) \right)^{a} \prod_{k=a}^{b-1} (q^c - q^k)^a \\
    &\leq \left( \prod_{j=0}^{a-1} (q^c - q^j) \right)^{a} \prod_{k=a}^{b-1} (q^c - q^a)^a \\
    &<  \left( \prod_{j=0}^{a-1} (q^c - q^j) \right)^{a} \prod_{k=a}^{b-1} \left( \prod_{j=0}^{a-1} (q^c - q^j) \right) =  \left( \prod_{j=0}^{a-1} (q^c - q^j) \right)^{b}. \\
\end{align}
Hence, we observe that
\begin{align}\label{eq: gamma q m eta}
 \prod_{j=0}^{i-1} (q^\mathfrak{M} - q^j)  > \left( \prod_{j=0}^{\mu-1} (q^\mathfrak{M} - q^j) \right)^{i/\mu}
= q^{i \, \mathfrak{M}} \left({\gamma_{q,m,\eta}}^{-1}\right)^{i/\mu} ,
\end{align}
where we introduce the notation
\[\label{eq: gamma_q_m_eta}
{\gamma_{q,m,\eta}} := \prod_{j=\mathfrak{M}-\mu+1}^{\mathfrak{M}} \left(1 - q^{-j}\right)^{-1}
\]
(with $\mu = \mu_{\Sigma R} = \min\{m,\eta\}$ and $\mathfrak{M} := \max\{m,\eta\}$). Combining \eqref{eq: 1q2 leq q} and \eqref{eq: gamma q m eta} leads to a new lower bound on the number of matrices of rank $i$.
\begin{proposition}\label{prop: new gamma LB number matrices} 
    For $m, \eta, i \in \N$ with $0 \leq i \leq \mu$, we have the lower bound
    \begin{align}
        \left({\gamma_{q,m,\eta}}^{-1/\mu}\right)^i \begin{bmatrix}\mu\\i\end{bmatrix}_{1/q^2} q^{i(m + \eta - i)}\leq \NM_q(m,\eta, i).
    \end{align}
\end{proposition} 
\begin{proof}\renewcommand{\qedsymbol}{}
    Let $\mathfrak{M} = \max\{m,\eta\}$. Then
    \begin{align}
      \text{}\hspace{10mm}\begin{bmatrix}\mu\\i\end{bmatrix}_{1/q^2} q^{i(m + \eta - i)} \left({\gamma_{q,m,\eta}}^{-1/\mu}\right)^i &= 	 \begin{bmatrix}\mu\\i\end{bmatrix}_{1/q^2} q^{i(\mu - i)}  q^{i\mathfrak{M}} \left({\gamma_{q,m,\eta}}^{-1/\mu}\right)^i \\
    	&\leq    \begin{bmatrix}\mu\\i\end{bmatrix}_{q}   q^{i\mathfrak{M}} \left({\gamma_{q,m,\eta}}^{-1}\right)^{i/\mu}\\
    	&<  \begin{bmatrix}\mu\\i\end{bmatrix}_{q}  \prod_{j=0}^{i-1}  \left(q^\mathfrak{M}-q^{j}\right) = \NM_q(m,\eta, i) \hspace{11mm}\square
    \end{align}
\end{proof}
\vspace{-5mm}
Next, we observe that 
\begin{align}
\NM_q(m,\eta, t) = \left(\prod_{i=1}^{t}\frac{(1-q^{-m+i-1})(1-q^{-\eta+i-1})}{(1-q^{-i})}\right) q^{t(m+\eta-t)}.
\end{align}
Let us now introduce the function
\[
\kappa_{q, m, \eta}(t) := \left(\frac{(1-q^{-m})(1-q^{-\eta})}{(1-q^{-1})}\right)^{t}.
\]
Noticing that $\kappa_{q, m, \eta}(t) \leq \prod_{i=1}^{t}\frac{(1-q^{-m+i-1})(1-q^{-\eta+i-1})}{(1-q^{-i})}$, we  derive the following upper bound on $\NM_q(m,\eta, t)$.

\begin{proposition}\label{prop: upper bound C2 on number matrices} 
    For $m, \eta, t \in \N$ we have the upper bound
    
    \begin{align}
        \NM_q(m,\eta, t) \ \leq \ \kappa_{q, m, \eta}(t) q^{t(m+\eta-t)}.
    \end{align}
\end{proposition}
With the property $\kappa_{q, m, \eta}(t_1)\kappa_{q, m, \eta}(t_2) = \kappa_{q, m, \eta}(t_1+t_2)$, we perform the same steps of the proof of bounds \eqref{eq: upper  bound sven} in \cite{puchinger2022generic} with $\kappa_{q, m, \eta}(t)$ instead of $\gamma_q$ and obtain Theorem \ref{thm:improvedSven} as \textcolor{review}{a} new upper bound on sum-rank sphere sizes.
\begin{theorem}\label{thm:improvedSven} 
    Given positive integers $m, \eta, \ell, t$ and a prime power $q$, we have
    \begin{align}
        \card{\sphere^{m,\eta,\ell,q}_t} \textcolor{review}{ \leq \kappa_{q, m, \eta}(t) \binom{\ell + t - 1}{\ell - 1} q^{t(m+\eta-\frac{t}{\ell})+\frac{r^2}{\ell}-r}} \leq \kappa_{q, m, \eta}(t) \binom{\ell + t - 1}{\ell - 1} q^{t(m+\eta-\frac{t}{\ell})}
    \end{align}
\end{theorem}
\noindent
\textcolor{review}{where $r$ satisfies $t \equiv r \mod \ell$ with $0  \leq r < \ell$.}
Importantly, for sufficiently small $t$ the coefficient $\kappa_{q, m, \eta}(t)$ is smaller than $\gamma_q$. In this case the bounds of Theorem \ref{thm:improvedSven} are tighter than \cite[Theorem 5]{puchinger2022generic} (see \eqref{eq: upper  bound sven}).

\subsection{Convolution Upper Bound}\label{subsec: upper bound convolution}
Let $f(x)$ and $g(x)$ be two real-valued functions defined on the natural numbers (or on a larger domain). For $t \in \N$, we define the \textbf{discrete convolution} by 
\begin{align}
     [f *  g](t) := \sum_{i = 0}^t f(i)g(t-i).
\end{align}
The $\ell$-fold discrete convolution $[f * f * \cdots  * f]$ (well-defined by associativity of $*$) is denoted as $f^{*\ell}$. 
The strategy in this subsection comes from the observation that \eqref{eq: exact_sphere_size_sum-rank} is equivalent to 
\[
\card{\sphere^{m,\eta,\ell,q}_t} = \NM_q(m,\eta,x)^{*\ell}(t). \vspace{2mm}
\]
Let $C(t)$ be a real-valued function depending on the parameters $m,\eta,q$ and satisfying
\begin{align}
    \NM_q(m,\eta,t) \leq C(t) q^{t(m + \eta - t)} \quad \text{and} \quad C(t_1)  C(t_2) = C(t_3) C(t_4)
\end{align}
whenever $t_1 + t_2 = t_3 + t_4$.  By Proposition \ref{prop: upper bound C2 on number matrices}, examples of such functions are $\gamma_q$ and $\kappa_{q, m, \eta}(t)$.
The reason for looking at these functions is because they work well with discrete convolutions, i.e.,
    $\left[C(x)f(x) * C(x)g(x)\right](t) = C(0)C(t)[f * g](t).$ 
Therefore, we can upper bound the sphere sizes as follows
\begin{align}\label{eq: upper bound convolution 1}
\card{\sphere^{m,\eta,\ell,q}_t}
     \leq \left(C(x)q^{x(m + \eta - x)}\right)^{*\ell}(t)
     = C(0)^{\ell-1}C(t)\left(q^{x(m + \eta - x)}\right)^{*\ell}(t).
\end{align}
Proposition \ref{prop: Bounding sum by integral} provides a formula to easily bound convolutions.
\begin{proposition}\label{prop: Bounding sum by integral} 
Consider $f_\ell(x) := q^{x(m+\eta - x/\ell)}$ for $x \in \R$ and $\ell \in \N$. Functions of this form satisfy the convolution inequality 
\begin{align}
    [f_{\ell_1} * f_{\ell_2}](t) \leq  \left(1 + \sqrt{\frac{\ell_1 \ell_2\pi}{(\ell_1 + \ell_2) \ln{q}}}\right)f_{\ell_1 + \ell_2}(t).
\end{align}
\end{proposition}
\vspace{-0.2cm}
\begin{proof}
First of all, note that 
\begin{align}
[f_{\ell_1} * f_{\ell_2}](t) 
&= \sum_{i=0}^{t} q^{t(m+\eta - \frac{t}{\ell_1+\ell_2}) - \frac{1}{\ell_1}i^2 - \frac{1}{\ell_2}(t-i)^2 + \frac{t^2}{\ell_1+\ell_2} } \\
&= f_{\ell_1 + \ell_2}(t) \sum_{i=0}^{t}q^{-\left(\frac{1}{\ell_1} + \frac{1}{\ell_2}\right)\left(i - \frac{\ell_1}{\ell_1 + \ell_2}t\right)^2}  \textcolor{review}{< f_{\ell_1 + \ell_2}(t)\!\!\! \sum_{i=-\infty}^{\infty} \!\!q^{-\left(\frac{1}{\ell_1} + \frac{1}{\ell_2}\right)\left(i - \frac{\ell_1}{\ell_1 + \ell_2}t\right)^2}\!\!.}
\end{align}
Next, we will apply the technique of bounding summations with integrals, bounding the discrete convolution with a continuous one. To make the notation more compact, let $g(x):= q^{-\left(\frac{1}{\ell_1} + \frac{1}{\ell_2}\right)\left(x - \frac{\ell_1}{\ell_1 + \ell_2}t\right)^2}$. Then, using unimodality of $g(x)$ with its mode/peak at $x_{\text{p}} := \frac{\ell_1}{\ell_1 + \ell_2}t$ where $g$ attains $g(x_{\text{p}}) = 1$, we  bound different parts of the summation as follows:
\begingroup
\addtolength{\jot}{1em}
\textcolor{review}{\begin{align}
\sum_{i=-\infty}^{\left\lfloor x_{\text{p}}\right\rfloor -1} g(i)  \; \leq \; \int_{-\infty}^{\left\lfloor x_{\text{p}}\right\rfloor} g(x)  \ dx,& \qquad \qquad \sum_{i = \left\lceil x_{\text{p}}\right\rceil -1}^\infty g(i)  \; \leq \; \int_{\left\lceil x_{\text{p}}\right\rceil}^\infty g(x) \ dx, \\ 
\sum_{i=\left\lfloor x_{\text{p}}\right\rfloor }^{\left\lceil x_{\text{p}}\right\rceil} g(i) & \; \leq \;  1 + \int_{\left\lfloor x_{\text{p}}\right\rfloor}^{\left\lceil x_{\text{p}}\right\rceil} g(x) \ dx. 
\end{align}}
\endgroup
Note that these bounds hold both if $x_{\text{p}} \in \N$ and $x_{\text{p}} \notin \N$.
Hence,
\begin{align}
    \sum_{i=-\infty}^{\infty}q^{-\left(\frac{1}{\ell_1} + \frac{1}{\ell_2}\right)\left(i - \frac{\ell_1}{\ell_1 + \ell_2}t\right)^2}  &\leq  1 + \int_{-\infty}^{\infty}q^{-\left(\frac{1}{\ell_1} + \frac{1}{\ell_2}\right)\left(x - \frac{\ell_1}{\ell_1 + \ell_2}t\right)^2} dx \\ 
    &=  1 + \int_{-\infty}^{\infty}q^{-\left(\frac{1}{\ell_1} + \frac{1}{\ell_2}\right)x^2} dx = 1+\sqrt{\frac{\pi}{\left(\frac{1}{\ell_1} + \frac{1}{\ell_2}\right) \ln{q}}}
\end{align}
and thus the proposition follows.\hfill
\end{proof}

\textcolor{review}{A tighter bound could possibly be obtained by examining the summation $\sum_{i=-\infty}^{\infty}q^{-\left(\frac{1}{\ell_1} + \frac{1}{\ell_2}\right)\left(i - \frac{\ell_1}{\ell_1 + \ell_2}t\right)^2}$ more closely, which is a special case of a \textit{(Jacobi) theta function}. Studying theta functions lies outside of the scope of this paper, but could be a promising direction for further improvements on Proposition \ref{prop: Bounding sum by integral}.}

Setting $\ell_1 = 1$ and applying Proposition \ref{prop: Bounding sum by integral} inductively for $\ell_2 = 1,\ldots,\ell-1$, we obtain new upper bounds on the sum-rank sphere sizes using \eqref{eq: upper bound convolution 1} with $C(t)$ of our choice.
\begin{theorem}\label{thm: integral UB}
Let $m, \eta, \ell, q, t$ be positive integers. Choosing $C(t)$ equal to $\gamma_q$ or $\kappa_{q, m, \eta}(t)$, we observe the following bounds, respectively
\begin{align}
    \card{\sphere^{m,\eta,\ell,q}_t} &\leq \gamma_q^\ell \ \prod_{k=1}^{\ell-1}\left(1 + \sqrt{\frac{k\pi}{(k+1) \ln{q}}}\right) q^{t(m+\eta - t/\ell)}\\
    \card{\sphere^{m,\eta,\ell,q}_t} &\leq \kappa_{q, m, \eta}(t) \prod_{k=1}^{\ell-1}\left(1 + \sqrt{\frac{k\pi}{(k+1) \ln{q}}}\right)  q^{t(m+\eta - t/\ell)}
\end{align}    
where the further simplifications 
\[
\prod_{k=1}^{\ell-1}\left(1 + \sqrt{\frac{k\pi}{(k+1) \ln{q}}}\right) \leq \left(1 + \sqrt{\frac{(\ell-1)\pi}{\ell\ln{q}}} \right)^{\ell-1} < 
\left(1 + \sqrt{\frac{\pi}{ \ln{q}}} \right)^{\ell-1}
\]
can be made.
\end{theorem}

\subsection{Lower Bound via Ordinary Generating Functions}\label{subsec: lower bound OGF}
Recall that the OGF of sum-rank sphere sizes is given by
\begin{align}\label{eq:sumranksphere_size}
\sphere^{m,\eta,\ell,q}(z) = \sum_{t=0}^{\mu\ell} \card{\sphere_t^{m,\eta,\ell,q}} \ z^t = \sphere^{m,\eta,1,q}(z)^\ell.
\end{align}
Note that finding a closed-form bound for $\card{\sphere_t^{m,\eta,\ell,q}}$ is equal to finding such a bound for the coefficients of the $\ell$-th power on the right-hand side of \eqref{eq:sumranksphere_size}.
The approach that we thus take is, to bound the function $\sphere^{m,\eta,1,q}(z) = \sum_{i=0}^\mu \NM_q(m,\eta,i) \, z^i$ coefficient-wise using another polynomial $\mathcal{F}(z)$ whose $\ell$-th power can be computed more easily, as we will see later. We consider the polynomial
\[
\mathcal{F}(z):=\sum_{i=0}^{\mu} q^{i(m+\eta-i) } \begin{bmatrix}\mu\\i\end{bmatrix}_{1/q^2} z^i. 
\]
This polynomial satisfies the following chain of coefficient-wise inequalities:
\begin{align}
\sum_{i=0}^{\mu} \gamma_q^{-1} q^{i(m+\eta -i)}z^i \ \preccurlyeq_c \ \gamma_q^{-1} \mathcal{F}(z)
\ \preccurlyeq_c \ \mathcal{F}({\gamma_{q,m,\eta}}^{-1/\mu} z) \ \preccurlyeq_c \ \sphere^{m,\eta, 1, q}(z).
\end{align}
    The first inequality follows from $ 1 \leq \begin{bmatrix}\mu\\i\end{bmatrix}_{1/q^2}$, the second inequality follows from $\gamma_q^{-1} \leq {\gamma_{q,m,\eta}}^{-1} \leq {\gamma_{q,m,\eta}}^{-i/\mu} \leq 1$ for $0 \leq i \leq \mu$ and the third from Proposition \ref{prop: new gamma LB number matrices}.
    Since these coefficient-wise inequalities are preserved under $\ell$-th powers, we obtain
    \begin{align}\label{eq:chain two inequalities}
    \left(\sum_{i=0}^{\mu} \gamma_q^{-1} q^{i(m+\eta -i)}z^i\right)^\ell \ \preccurlyeq_c  \ \mathcal{F}({\gamma_{q,m,\eta}}^{-1/\mu} z)^\ell \ \preccurlyeq_c \ \sphere^{m,\eta, \ell, q}(z).
    \end{align}

Note that the left-most polynomial is used in deriving the lower bound  \cite[Lemma 2]{Ott2021BoundsCodes} (see  \eqref{eq: lower  bound ott}). Since the middle polynomial is coefficient-wise greater, estimating its coefficients can potentially lead to improved lower bounds. Let us thus focus on estimating the $\ell$-th power of $\mathcal{F}({\gamma_{q,m,\eta}}^{-1/\mu} z)$.

To our advantage, the polynomial $\mathcal{F}(z)$ (and hence also $\mathcal{F}({\gamma_{q,m,\eta}}^{-1/\mu} z)$) can be factored nicely into linear parts, as the following proposition shows.

\begin{proposition} Let $m,\eta \in \N$ and $\mu = \min\{m,\eta\}$. Then
\[
\mathcal{F}(z) =\sum_{i=0}^{\mu} q^{i(m+\eta-i) } \begin{bmatrix}\mu\\i\end{bmatrix}_{1/q^2} z^i = \prod_{i=1}^{\mu} (1 + q^{m+\eta-2i + 1} z).
\]
\end{proposition}\label{prop: q binomial theorem}
\noindent
The proof follows directly from  the \textit{$q$-binomial theorem}, also known as the \textit{Cauchy binomial theorem} or \textit{Gaussian binomial theorem} (see \cite[Chapter 5]{kac2002quantum}).

Let us now consider  
\begin{align}
    \mathcal{F}(z)^\ell &= \prod_{i=1}^{\mu} \left(1 + q^{-2i} \left(q^{m+\eta + 1} z\right)\right)^\ell \\ 
    &= \prod_{i=1}^{\mu} \, \sum_{j = 0}^\ell \binom{\ell}{j} q^{-2ij} \left(q^{m+\eta + 1} z\right)^j  \\
    &= \sum_{t=0}^{\mu \ell} \left( \sum_{\substack{j_1  + \ldots + j_{\mu} = t \\ 0 \leq j_i \leq \ell}} \; \prod_{i=1}^{\mu}\binom{\ell}{j_i} q^{-2i j_i} \right) \left(q^{m+\eta + 1} z\right)^t.
\end{align}
At first, bounding these sums seems as difficult as those in  \eqref{eq: exact_sphere_size_sum-rank}. However, unlike \eqref{eq: exact_sphere_size_sum-rank},  the factors $\binom{\ell}{j_i} q^{-2ij_i}$ \textcolor{review}{are} also strongly dependent on  $i$ (not only on $j_i$) and increase as $i$ decreases.
 This implies that there is a \textit{unique} optimal tuple $(j_1,\ldots,j_\mu)$ that maximizes $\prod_{i=1}^{\mu}\binom{\ell}{j_i} q^{-2ij_i}$, in contrast to  \eqref{eq: exact_sphere_size_sum-rank} where every tuple $(t_1,\ldots, t_\ell)$ contributes the same term $ \prod_{i=1}^\ell \NM_q(m,\eta, t_{i})$ to the sum as any of its permutations.
 
For each $t$ we will simply bound the sum by only one term $\prod_{i=1}^{\mu}\binom{\ell}{j_i} q^{-2ij_i}$ with the near-optimal choice
\begin{align}\label{eq:cases_sum}
    j_i =
 \begin{cases}
     \ell &  \text{for }\ 1 \leq i \leq t_* \\
     r & \text{for }\ i = t_*+1  \\
    0 &  \text{for }\ t_*+2 \leq i \leq \mu 
    \end{cases}
\end{align}
where $t = t_* \ell + r$ for some $t_* \in \N$ and $0  \leq r < \ell$. This yields the bound
$$
\mathcal{F}(z)^\ell 
\succcurlyeq_c \sum_{t=0}^{\mu \ell} 
\left(\binom{\ell}{r}q^{- 2r(t_* + 1)}\prod_{i=1}^{t_*} q^{-2i\ell}\right) \left(q^{m+\eta + 1} z\right)^t
=
\sum_{t=0}^{\mu \ell} \binom{\ell}{r} q^{t(m+\eta - \frac{t}{\ell})+\frac{r^2}{\ell}-r} z^t.
$$
Finally, by substituting ${\gamma_{q,m,\eta}}^{-1/\mu} z$ for $z$ in this inequality and applying \eqref{eq:chain two inequalities} we obtain the following result.
\begin{theorem}\label{thm:lowerbound}  Let $m,\eta \in \N$, $\mu = \min\{m,\eta\}$ and $q$ a prime power. Let $t = t_* \ell + r \leq \mu \ell$ with $t_* \in \N$ and $0  \leq r < \ell$. Then 
    \begin{align}
\left({\gamma_{q,m,\eta}}^{-1}\right)^{t/\mu}  \binom{\ell}{r} q^{t(m+\eta - \frac{t}{\ell})+\frac{r^2}{\ell}-r} \leq \card{\sphere_t^{m,\eta,\ell,q}},
    \end{align}
with $\gamma_{q,m,\eta}$ given in  \eqref{eq: gamma_q_m_eta}.
\end{theorem}
Notice that remarkably, aside for the improved coefficients in front, we have obtained the same lower bound as {\cite[Lemma 2]{Ott2021BoundsCodes} (see  \eqref{eq: lower  bound ott}) via a completely different method. However, by choosing $j_i$ differently in \eqref{eq:cases_sum} there is still room for future optimization.

In addition, for $0 < i < \mu$ we have
    \begin{align}
         \frac{\NM_q(m,\eta, i)^2}{\NM_q(m,\eta, i-1)\NM_q(m,\eta, i+1)} &= \frac{(q^m-q^{i-1})}{(q^m-q^i)}\frac{(q^\eta-q^{i-1})}{(q^\eta-q^i)}\frac{q^{i}(q^{i+1}-1)}{q^{i-1}(q^i-1)}\geq q^2 > 1
     \end{align}
     and thus $(\NM_q(m,\eta, i))_{i=0}^{\mu}$ is a \textit{logaritmically concave sequence} (\textit{log-concave} for short).
    Moreover, since convolution preserves log-concavity, it holds that for all $\ell$ the sequence
    $\left(\card{\sphere_i^{m,\eta,\ell,q}}\right)_{i=0}^{\mu\ell}$ is log-concave. Hence we can take the \textit{convex hull} of the logarithm of\\ 
    the bound in Theorem \ref{thm:lowerbound} to obtain a slight improvement. In this context, the convex hull of a sequence $(a_i)_{i=0}^{\mu \ell}$ is the sequence  $(b_i)_{i=0}^{\mu \ell}$ with 
    $$
    b_i = \max_{x,y > 0}\left\{a_i, \, \frac{y}{x+y}a_{i-x} + \frac{x}{x+y}a_{i+y}\right\},
    $$
    and as \textcolor{review}{a} consequence of Carathéodory's theorem on convex hulls this is the lowest-valued concave sequence that is coefficient-wise greater or equal to $(a_i)_{i=0}^{\mu \ell}$. 
\begin{theorem}\label{thm: convex hull lower bound}
The convex hull of the sequence 
$$\left(\log_q\left( \left({\gamma_{q,m,\eta}}^{-1}\right)^{i/\mu}  \binom{\ell}{r} q^{i(m+\eta - \frac{i}{\ell})+\frac{r^2}{\ell}-r}\right)\right)_{i=0}^{\mu\ell}$$ is a lower bound on $\log_q\card{\sphere_t^{m,\eta,\ell,q}}$.
\end{theorem}

\section{Comparison of Bounds}\label{sec:comparison}
In this section, we compare the new bounds presented in this paper with the existing bounds related to the sphere size in the sum-rank metric. 
In Figure~\ref{fig:sphere-size-vs-t-v2} the relationship between the growth rate $\frac{1}{\ell}\log_q \card{ \sphere_{t}^{m, \eta, \ell, q}}$ of the sphere size and the normalized radius $\rho$ is shown.
We observe that the upper bound using Theorem~\ref{thm:upperbound_sphere_saddlepoint} and the lower bound using Theorem~\ref{thm:entropy LB 2} are the tightest bounds and very close to the exact values.
The computation of these bounds necessitates the evaluation of the entropy $H_\rho$. Computing $H_\rho$ is straightforward for a specified $\beta$, whereas determining $\beta$ for a given $\rho$ cannot in general be achieved in a closed-form manner, as outlined in~\eqref{eq:weightconstraint}.

For scenarios where prioritizing closed-form expressions dependent on $\rho$ is essential, the derived alternative bounds may better suit the intended use-cases. In  Figure \ref{fig:sphere-size-vs-t-v2}, the upper bounds from Theorem~\ref{thm: integral UB} using $\kappa_{q, m, \eta}$, Theorem \ref{thm: integral UB} using $\gamma_q$ and Theorem~\ref{thm:improvedSven} are consolidated into a single piece-wise function by selecting the minimum value among these bounds. The transition points are indicated by circles.

We observe that for the new closed-form upper and lower bounds we improve significantly in comparison to the already existing closed-form bounds given in~\cite[Theorem 5]{puchinger2022generic} and~\cite[Lemma 2]{Ott2021BoundsCodes}. Furthermore, the new bounds are potentially useful tools for obtaining improved closed-form Gilbert-Varshamov or sphere-packing bounds, as introduced in~\cite{byrne2021fundamental} and~\cite{Ott2021BoundsCodes}.

\begin{figure}[H]
    \centering
    \begin{tikzpicture}
    \begin{axis}[
	width=\columnwidth,
	height=0.9\columnwidth,
	grid=both,
	grid style={dotted,gray},
	legend cell align=left,
	legend style={font=\small},
	legend pos=south east,
	mark options={solid},
	mark size=3,
        xmin = 0,
        ymin = 0,
        xmax = 5,
        ymax = 26,
        x filter/.code={\pgfmathparse{#1*1}\pgfmathresult},
	xlabel={Normalized radius $\rho$},
	ylabel={$\frac{1}{\ell}\log_q \card{ \sphere_{\rho\ell}^{m, \eta, \ell, q}}$}
	]
   
    \addplot[plotcyan,line width = 1pt, densely dash dot] table[x=rho,y=UB] {./TikzGraphs/UBjulian_q=2_m=5_eta=5_l=100.txt};\addlegendentry{\scriptsize  \cite[Theorem 5]{puchinger2022generic}, cf. \eqref{eq: upper  bound sven}};

    \addplot[plotred,line width = 1pt, dash dot dot] table[x=rho,y=UB] {./TikzGraphs/UBjulianAlt_q=2_m=5_eta=5_l=100.txt};\addlegendentry{\scriptsize  Theorem \ref{thm:improvedSven} }; 
    
    \addplot[plotred,line width = 1pt, densely dashed] table[x=rho,y=UB] {./TikzGraphs/UBmy5_q=2_m=5_eta=5_l=100.txt};\addlegendentry{\scriptsize  Theorem \ref{thm: integral UB} using $\kappa_{q, m, \eta}$};
    \addplot[plotred,line width = 1pt, dotted] table[x=rho,y=UB] {./TikzGraphs/UBmy4_q=2_m=5_eta=5_l=100.txt};\addlegendentry{\scriptsize  Theorem \ref{thm: integral UB} using $\gamma_q$};

   \filldraw[color=plotred] (0.71, 8.86) circle (1.7pt);
   \filldraw[color=white] (0.71, 8.86) circle (1pt);
    \filldraw[color=plotred] (1.98, 19.29) circle (1.7pt);
    \filldraw[color=white] (1.98, 19.29) circle (1pt);
    
    \addplot[red,line width = 1.5pt, loosely dotted] table[x=rho,y=UB] {./TikzGraphs/EntropyUB_q=2_m=5_eta=5.txt};\addlegendentry{\scriptsize  Theorem \ref{thm:upperbound_sphere_saddlepoint}};

    \addplot[plotyellow,line width = 1pt, solid] table[x=rho,y=UB] {./TikzGraphs/Actual_q=2_m=5_eta=5_l=100.txt};\addlegendentry{\scriptsize  Exact value};

    \addplot[plotblue,line width = 1pt, loosely dashed] table[x=rho,y=UB]{./TikzGraphs/EntropyLB_q=2_m=5_eta=5_l=100_eps=001.txt};\addlegendentry{\scriptsize  Theorem \ref{thm:entropy LB 2} with $\varepsilon = 0.01$};

    \addplot[plotgreen,line width = 0.5pt, solid] table[x=rho,y=UB] {./TikzGraphs/LBconvex_q=2_m=5_eta=5_l=100.txt};\addlegendentry{\scriptsize  Theorem \ref{thm: convex hull lower bound}};
    \addplot[plotgreen,line width = 1pt, densely dotted] table[x=rho,y=UB] {./TikzGraphs/LBmy_q=2_m=5_eta=5_l=100.txt};\addlegendentry{\scriptsize  Theorem \ref{thm:lowerbound}};
    
    \addplot[plotpurple,line width = 1pt, loosely dash dot] table[x=rho,y=UB] {./TikzGraphs/LBott_q=2_m=5_eta=5_l=100.txt};\addlegendentry{\scriptsize \cite[Lemma 2]{Ott2021BoundsCodes}, cf. \eqref{eq: lower  bound ott}};
 
    \end{axis}
\end{tikzpicture}
    \caption{Comparison of upper and lower bounds for the sphere $\sphere_{\rho\ell}^{m, \eta, \ell, q}$ as function of $\rho$ with parameters $q=2$, $m=5$, $\eta=5$, $\ell=100$.}
    \label{fig:sphere-size-vs-t-v2}
\end{figure}
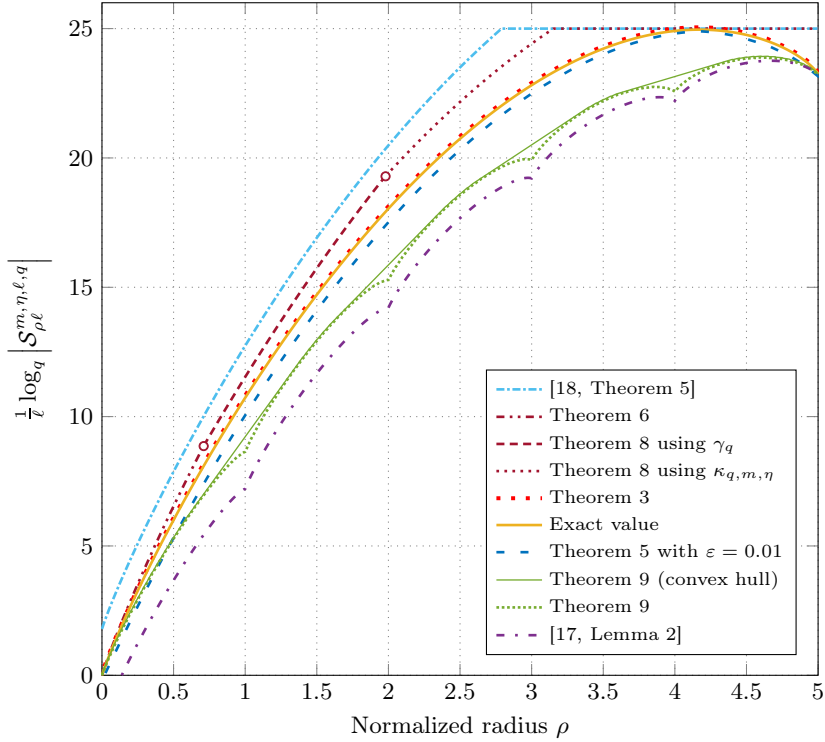

In Figure~\ref{fig:sphere-size-vs-L} we show the tightness of the improved bounds for different numbers of blocks. We choose the same values for the parameters $q$, $m$, $t$ and $n$ as for the bounds given in~\cite{puchinger2022generic}. Notably, the bounds proposed in~\cite{puchinger2022generic} exhibit considerable looseness in scenarios where $\ell$ becomes substantially large (i.e., when the sum-rank metric converges to the Hamming metric). While superior bounds are already established for the Hamming metric (i.e., $\ell=n$), our analysis illustrates substantial enhancements for $\ell < n$ compared to existing bounds.
\newpage

\begin{figure}
    \centering    
    \begin{tikzpicture}
    \begin{axis}[
	width=\columnwidth,
	height=0.9\columnwidth,
	grid=both,
	grid style={dotted,gray},
	legend cell align=left,
	legend style={font=\small,at={(0.5,0.97)}},
        legend pos=north east,
	mark options={solid},
	mark size= 0,
        xmin = 1,
        ymin = 300,
        ymax = 900,
        xmax = 60,
        xtick={1,2,3,4,5,6,10,12,15,20,30,60},
    xticklabels={\scalebox{0.9}{1},\scalebox{0.9}{2},\scalebox{0.9}{3},\scalebox{0.9}{4},\scalebox{0.9}{5},\scalebox{0.9}{6},\scalebox{0.9}{10},\scalebox{0.9}{12},\scalebox{0.9}{15},\scalebox{0.9}{20},\scalebox{0.9}{30},\scalebox{0.9}{60}},
	xlabel={Number of blocks $\ell$},
	ylabel={$\log_q \card{ \sphere_{t}^{m, \eta, \ell, q}}$},
	]
    
   \addplot[plotred,line width = 1pt, dotted, mark=*] table[x=t,y=value] {./TikzGraphs/LUBmy4_q=2_m=40_n=60_t=10.txt};\addlegendentry{\scriptsize  Theorem \ref{thm: integral UB} using $\gamma_q$};

    \addplot[plotcyan,line width = 1pt, densely dash dot, mark=*] table[x=t,y=value]{./TikzGraphs/Ljulian_q=2_m=40_n=60_t=10.txt};\addlegendentry{\scriptsize  \cite[Theorem 5]{puchinger2022generic}, cf. \eqref{eq: upper  bound sven}};

    \addplot[plotred,line width = 1pt, densely dashed, mark=*] table[x=t,y=value]{./TikzGraphs/LUBmy5_q=2_m=40_n=60_t=10.txt};\addlegendentry{\scriptsize  Theorem \ref{thm: integral UB} using $\kappa_{q, m, \eta}$};
 
    \addplot[plotred,line width = 1pt, dash dot dot, mark=*] table[x=t,y=value] {./TikzGraphs/LjulianAlt_q=2_m=40_n=60_t=10.txt};\addlegendentry{\scriptsize  Theorem \ref{thm:improvedSven} }; 
    
    \addplot[plotyellow,line width = 1pt, solid, mark=*] table[x=t,y=value]{./TikzGraphs/Lactual_q=2_m=40_n=60_t=10.txt};\addlegendentry{\scriptsize  Exact value};

    \addplot[plotgreen,line width = 1pt, densely dotted, mark=*] table[x=t,y=value] {./TikzGraphs/LLBmy_q=2_m=40_n=60_t=10.txt};\addlegendentry{\scriptsize  Theorem \ref{thm:lowerbound}};
    
    \addplot[plotpurple,line width = 1pt, loosely dash dot, mark=*] table[x=t,y=value] {./TikzGraphs/LLBottBetter_q=2_m=40_n=60_t=10.txt};\addlegendentry{\scriptsize  \cite[Lemma 2]{Ott2021BoundsCodes}, cf. \eqref{eq: lower  bound ott}};
    
	\end{axis}
\end{tikzpicture}
    \caption{Comparison of upper and lower bounds for the sphere $\sphere_{t}^{m, \eta, \ell, q}$ as function of $\ell$ with parameters $q=2$, $m=40$, $t=10$ and keeping $n= \eta \ell=60$ constant.}
    \label{fig:sphere-size-vs-L}
\end{figure}
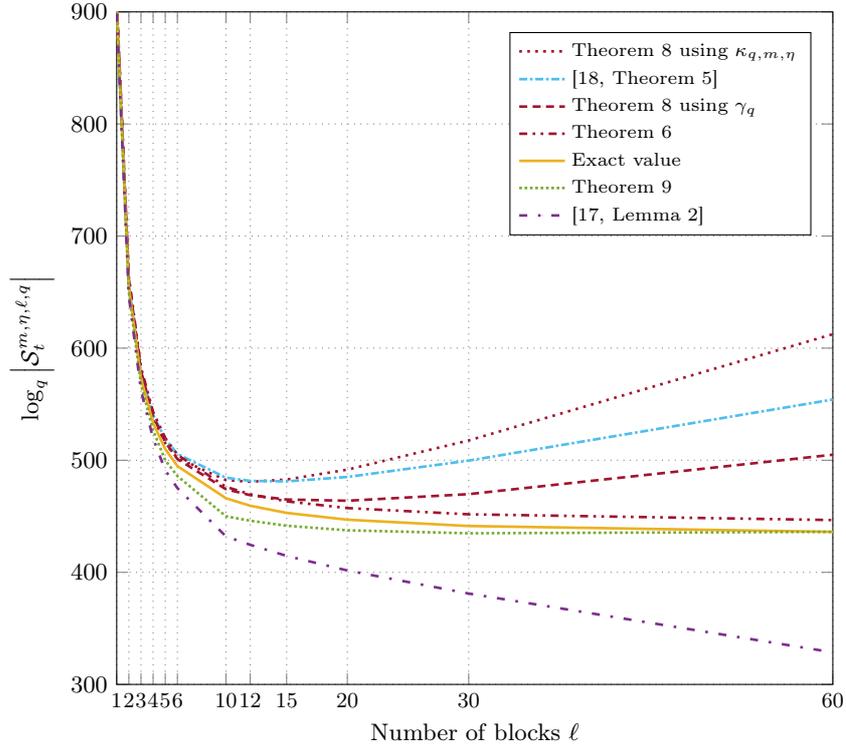

\section{Conclusions}\label{sec:conclusions}

In this paper, we derive novel upper and lower bounds on the size of an $\ell$-dimensional sphere. First, we consider a general setting for any coordinate-additive metric, and then we focus specifically on spheres in the sum-rank metric. 
In the general setting, we employ techniques from information theory and typical sequences to bound the logarithmic growth rate of the sphere's size, utilizing the entropy of a Boltzmann-like distribution representing the marginal distribution of a typical sequence. 
Although these bounds apply to any coordinate-additive metric, we also derive new bounds specifically for the sum-rank metric, using convolution arguments for the upper bounds and ordinary generating functions for the lower bounds. 
Finally, we compare our bounds to existing ones in the sum-rank metric and observe that, when choosing the most suitable bound for each regime, our bounds outperform the existing closed-form bounds.

\bmhead{Acknowledgments}
H.~Sauerbier Couvée is supported by the European Research Council (ERC) under the European Union’s Horizon 2020 research and innovation programme (Grant agreement No. 801434) and the Bavarian State Ministry of Science and Arts via the project EQAP. T.~Jerkovits and J.~Bariffi acknowledge the financial support by the Federal Ministry of Education and Research of Germany in the program of "Souver\"an. Digital. Vernetzt." Joint project 6G-RIC, project identification number: 16KISK022. J.~Bariffi is funded by the European Union (DiDAX, 101115134). Views and opinions expressed are however those of the author(s) only and do not necessarily reflect those of the European Union or the European Research Council Executive Agency. Neither the European Union nor the granting authority can be held responsible for them.
The authors thank H.~Bartz, A.~Baumeister, G.~Liva, and A.~Wachter-Zeh for their insights and discussions.

\bibliography{references}

\end{document}